\documentclass[preprint,12pt]{elsarticle}

\usepackage{amsmath}
\usepackage{amsthm}
\usepackage{amssymb}
\usepackage{multirow}
\usepackage[table]{xcolor}
\usepackage[ruled, vlined]{algorithm2e}

\usepackage{tikz}
\usetikzlibrary{shapes,arrows,shadows}
\usetikzlibrary{positioning}

\usepackage{chngcntr}
\counterwithin{figure}{section}

\usepackage{comment}
\usepackage{enumitem}

\makeatletter
\newcommand*{\rom}[1]{\expandafter\@slowromancap\romannumeral #1@}

\renewcommand{\b}[1]{\textbf{#1}}

\newtheorem{example}{Example}
\newtheorem{theorem}{Theorem}
\newtheorem{lemma}{Lemma}

\newtheorem{observation}{Observation}

\newcommand{\gapms}{{\footnotesize GAP-BS}}

\newcommand{\tie}{truthful}
\newcommand{\lpe}{LP[E]}
\newcommand{\bin}{bin}
\newcommand{\bins}{bins}
\newcommand{\job}{item}
\newcommand{\jobs}{items}

\newcommand{\E}{\mathbb{E}}

\journal{Operations Research Letters}

\begin{document}

\begin{frontmatter}



\title{Generalized Assignment Problem: Truthful Mechanism Design without Money}


\author{Salman Fadaei\corref{cor1}}
\author{Martin Bichler\corref{cor2}}

\address{Department of Informatics, Technical University of Munich, \\ salman.fadaei@gmail.com, bichler@in.tum.de}

\cortext[cor1]{Corresponding author}
\begin{abstract}
In this paper, we study a mechanism design problem for a strategic variant of the generalized assignment problem (GAP) in a both payment-free and prior-free environment.
In GAP, a set of items has to be optimally assigned to a set of bins without exceeding the capacity of any singular bin. 
In the strategic variant of the problem we study, bins are held by strategic agents, and each agent may hide its compatibility with some items in order to obtain items of higher values. 
The compatibility between an agent and an item encodes the willingness of the agent to receive the item.
Our goal is to maximize total value (sum of agents' values, or social welfare) while certifying no agent can benefit from hiding its compatibility with items. 
The model has applications in auctions with budgeted bidders.
\end{abstract}


\end{frontmatter}

\section{Introduction} \label{intro}


Truthful mechanism design without money under general preferences is a classic topic in social choice theory.
Truthfulness ensures that no agent can be better off by manipulating its true preferences.
When searching for truthful mechanisms \emph{without money}, one has to look at restricted domains of preferences. 
The reason for this, is the Gibbard-Satterthwaite theorem which states that any truthful social choice function which selects an outcome among three or more alternatives has to be trivially aligned with the preference of a single agent (namely, the dictator) \citep{Gibbard73,Satterthwaite75}.
Thus, exploring domains for which there exist truthful mechanisms is of central importance in the field of social choice theory.

As an example for restricted domains, consider agents with single-peaked preferences. In this domain returning the \textit{median} of the peaks determines a truthful social choice \citep{moulin1980strategy}.
Another example is the two-sided matching,  in which a set of men has a strict preference ordering over a set of women, and vice versa.
A matching is an assignment of men to women where each side is assigned to only one element of the other side. 
The \emph{deferred acceptance algorithm} finds a stable matching which is truthful for the proposing side, but not necessarily truthful for the other side \citep{roth1992two}.

One way to circumvent the impossibility result is relaxing the social choice function.
Procaccia and Tennenholtz introduced the technique of welfare approximation as a means to derive truthful approximation mechanisms without money \citep{procaccia2013approximate}. 
This type of approximation is not meant to handle computational intractability, but a method to achieve truthfulness by relaxing the goal of optimizing social welfare (approximating social welfare), and thus circumventing the Gibbard-Satterthwaite impossibility theorem. 
The approach is to maximize welfare without considering incentives, and refer to this as optimal value. 
Then it is said that a truthful mechanism returns (at most) an $\alpha$-approximation of the optimal if its value is always greater than or equal to $1/\alpha$ times the optimal value ($\alpha \geq 1$).
Several works, subsequent to the work of Procaccia and Tennenholtz, employ this technique \citep{dughmi2010truthful,chen2013truthful,koutsoupias2014scheduling}. 
We apply this technique to a novel strategic setting in the following.

\subsection{Model}
Consider a strategic variant of the generalized assignment problem termed \gapms\ in an environment which is both prior-free and payment-free.
In \gapms, there are $m$ items $J$ and $n$ bins (knapsacks) $I$. Each bin $i$ has a capacity $C_i$ and associates a value $v_{ij}$ and a size $w_{ij}$ to any item $j$. 
A feasible assignment may allocate a subset of items $S$ to bin $i$ such that $\sum_{j \in S}w_{ij} \leq C_i$.
A feasible assignment may assign each item at most once.

In \gapms, we assume tuple $T=(\{v_{ij}\}_{ij},\{w_{ij}\}_{ij},\{C_i\}_i)$ is public, but each bin is held by a strategic agent.
The private information that each agent/bin holds is the set of its compatible items. 
The compatibility between an agent and an item encodes the willingness of the agent to receive the item.
In particular, consider a bipartite graph $G$ where one side corresponds to items and the other side corresponds to bins. 
The edges of $G$, $E \subseteq I\times J$ represent the compatible item-bin pairs.
The private type of a bin $i$ is therefore the set of edges in the graph incident on $i$, i.e. $E_i$. 
A bin $i$ receives value $v_i(S)=\sum_{j\in S:(i,j)\in E_i} v_{ij}$ from package $S$ if $\sum_{j\in S} w_{ij}<C_i$ and $0$, otherwise.
The total value of a feasible assignment $(S_1,S_2,\ldots,S_n)$ equals the sum of values received by the bins from the assignment: $\sum_{i \in I} v_i(S_i)$.
We seek a total value-maximizing algorithm that provides each bin $i$ with incentives to truthfully report its compatible items $E_i$ rather than any $E'_i \subset E_i$.\footnote{In fact, our results certify that each bin $i$ reports exactly $E_i$, and has no incentives to report any other set of edges $E'_i$. However, for the sake of simplicity in the exposition of the results, we focus on untruthful reports that are made by hiding some edges, $E'_i \subset E_i$.}
In other words, given a truthful mechanism, bins have no incentive to hide their compatibility with some items.

Let $\mathcal{A}$ denote a randomized algorithm which takes instance $(T,E)$ and computes $X \in \{0,1\}^E$, an assignment of items to bins. 
Notice, the assignment itself is a deterministic assignment (each bin receives a deterministic set of items), 
but algorithm $\mathcal{A}$ is internally randomized, i.e., $\mathcal{A}$ returns a solution which is randomly chosen according to a probability distribution over feasible assignments.
Thus, the computed assignment may change by running $\mathcal{A}$, twice on the same input.
Randomized algorithm $\mathcal{A}$, given any tuple $T$, should satisfy the following properties.

\begin{enumerate} [label={\textit{\roman*}.}] 
\item (feasibility)  $\forall j \in J$, $Pr[\sum_{i\in I} X_{ij}\leq 1]=1$ and $\forall i \in I$, $Pr[\sum_{j \in J} w_{ij}X_{ij}\leq C_i]=1$, where $X \sim \mathcal{A}(T,E)$, for all $E$.

\item (incentive compatibility, or truthfulness)  for all $i$, $E_i$, $E_{-i}$, and any reported $E'_i \subset E_i$, we have $\E[\sum_{j:(i,j)\in E_i}v_{ij} X_{ij}] \geq \E[\sum_{j:(i,j)\in E_i}v_{ij}X'_{ij}]$, where $X \sim \mathcal{A}(T,E)$, and $X'\sim \mathcal{A}(T,E'_i\cup E_{-i})$. 
\end{enumerate}
$E_{-i}$ always denotes $E\setminus E_i$. 
The expectation in $ii$ is taken over the coin flips of the algorithm. 
Note that, the expected value of the bin in both cases is calculated with respect to true item-bin compatibilities, $E$.
We remark that condition $ii$ characterizes mechanisms that are dominant strategy incentive compatible. 
In this paper, for brevity, we refer to these mechanisms as truthful mechanisms or algorithms.
To sum, our objective is to propose a randomized algorithm $\mathcal{A}$ for \gapms\ which is truthful, and always returns a feasible assignment whose value approximates the optimal total value as high as possible.

Many real-world decision problems can be modeled by variants of knapsack problems, therefore we believe that our model can be applied broadly.
As an example, we refer to the  \emph{maximum budgeted allocations} (MBA) problem \citep{chakrabarty2010approximability}.
In MBA, a set of indivisible items has to be assigned to a set of bidders.
Each bidder $i$ reports her willingness to pay $b_{ij}$ for item $j$ by bidding for the item, while she has a budget constraint $B_i$.
Each bidder $i$ on receiving a package $S$ of items, pays $\sum_{j\in S} b_{ij}$. 
Each bidder $i$ has the rigid constraint $B_i$ on her payment.
The goal in MBA is to find a distribution of items among the bidders which maximizes the total revenue (the sum of the payments by the bidders while respecting their budget constraints).
MBA arises in auctions with budgeted bidders and has several applications \citep{chakrabarty2010approximability}.

In MBA, bidders want to get as much as they can without spending more than their budget.
For instance, advertisers wish to maximize the impressions, clicks, or sales generated by their advertising, subject to budget constraints.
Similarly, bidders who have no direct utility for leftover money (e.g. because the money comes from a corporate budget) will buy as much as possible.
This types of bidders are called \emph{value maximizers}, and have recently drawn the attention of researchers in mechanism design \citep{rad16valuebidders, fadaei2016value}.

Consider a strategic variant of MBA in which each bidder,  in order to obtain a more valuable package of items, strategizes in the following way. 
Each bidder may strategically hide her interests in buying some items by not bidding for those items.
In this setting, the auctioneer wishes to certify that each bidder truthfully reveals her willingness to buy items.
In other words, a truthful mechanism in this setting will encourage participation of the bidders in the auction.
We model this setting by \gapms\ in which each bidder is represented by a bin, budgets $B_i$ by capacities $C_i$, the bids $b_{ij}$ by the values of bins for the items $v_{ij}$, and the payment by a bidder $i$ for item $j$ by the weight of the item  on the bin, $w_{ij}$. 
Thus, in this setting of \gapms, we have $v_{ij}=w_{ij}$ for all $i$ and $j$.
For this problem, since the value density of each item is the same over all bins, we provide a truthful $4$-approximation algorithm.

\subsection{Discussion About the Assumptions} \label{subsec:discuss-assumptions}
Aside from the applications of the model discussed above, we emphasize that our assumptions (which imply a highly structured domain) are necessary to escape the impossibility results such as the Gibbard-Satterthwaite theorem and its variations \citep{barbera1990strategy}.
For example, we resort to welfare approximations because as stated by Theorem \ref{thm-low-bounds}, no deterministic (or randomized) algorithm whose value is optimal, exists for \gapms.
The lower bounds in Theorem \ref{thm-low-bounds} were  derived for a different setting with strategic items in the literature, however, we can reproduce and adapt the theorem for our setting.
\begin{theorem}\citep{dughmi2010truthful} \label{thm-low-bounds} 
No truthful deterministic algorithm with an approximation ratio better than $2$ exists for \gapms.
Moreover, no truthful-in-expectation randomized algorithm with an approximation ratio better than $1.09$ exists for \gapms.
\end{theorem}

Now, we consider a setting in which bins/agents have private values for items.
This setting is more general than \gapms\ in that, in this setting, the agents can manipulate their valuations for items.
This is in contrast to \gapms\ in which the agents can only hide their valuations for some items by hiding their compatibility with those items.
For this general setting, no deterministic (or randomized) truthful algorithm, with an interesting approximation ratio, exists.
To see this, consider a simple market with one item, and a set of agents.
This market is equivalent to the single-item auction, but without money.
We observe that no mechanism without money can find the (true) highest valuation for the item, as the agents can report arbitrarily high values for the item.
That is, no truthful algorithm can do any better than the algorithm which allocates the item to the bin which is uniformly chosen at random.
Such an algorithm provides a trivial approximation ratio of $1/n$, $n$ being the number of agents.

In a parallel setting, Dughmi et al. \cite{dughmi2010truthful} and Chen et al. \cite{chen2013truthful} studied GAP in an environment in which \emph{items} are held by strategic agents.
This is in contrast to our assumption that \emph{bins} are held by strategic agents. 
Hence, the solutions proposed by these authors are not directly applicable to \gapms.
In GAP each item can be assigned only once, thus the setting studied by Dughmi et al. is appropriate for modeling single-demand bidders who are interested in buying only a single item. However, our model analyzes strategic bins which can model multi-demand bidders, i.e., bidders who are interested in buying multiple items. In particular, the sizes in our model are at the side of strategic agents which properly models the bidders' budgets in MBA problem.

\subsection{Results and Technique}
In addition to \gapms, we also analyze two variants, namely the multiple knapsack problem in which each item has the same size and value over bins, 
and density-invariant GAP in which each item has the same value density (value per size) over the bins.

We observe that the relaxation and rounding technique is applicable to these problems.
The relaxation and rounding technique is a welfare approximation technique  \citep{procaccia2013approximate} based on linear programming relaxations.
To apply the technique, we start with a linear programming relaxation of the problem. 
Then, we need an algorithm which returns a fractional solution to the relaxation with an acceptable approximation ratio.
The algorithm has to be fractionally truthful, i.e., no agent can increase its fractional value by untruthful reports.
Finally, a rounding scheme which preserves truthfulness is applied to the fractional solution to obtain an integer solution. 
It should be noted that the relaxation and rounding technique has been previously applied to mechanism design without money in a different setting \citep{dughmi2010truthful}. 
This fact is realized in Theorem \ref{thm-frac-int-sol}.

We apply the technique successfully to our problems by proposing fractionally truthful algorithms with acceptable approximation ratios.
For the rounding scheme, we use a rounding method called \emph{randomized meta-rounding}, originally proposed by Carr and Vempala \cite{carr2000randomized}, and later applied by Lavi and Swamy \cite{lavi2011truthful} to mechanism design (with quasi-linear valuations).
Using the relaxation and rounding technique, for two variants of \gapms, the multiple knapsack problem, and density-invariant GAP, we propose truthful $4$-approximation algorithms. 
For \gapms, we show an $O(\ln{(U/L)})$-approximation mechanism  where $U$ and $L$ are the upper and lower bounds for value densities of the compatible item-bin pairs.

\section{Generalized Assignment Problem} \label{sec-gap}
We start with a linear programming relaxation of \gapms.
  \begin{align}
    \text{Maximize} \quad & \textstyle \sum_{i=1}^n \sum_{j=1}^m v_{ij}x_{ij}  \tag{\lpe} \\
	 \text{subject to} \quad & \textstyle \sum_{i=1}^n x_{ij} \leq 1 \quad \forall j \in J \notag\\
	 &\textstyle \sum_{j=1}^m w_{ij}x_{ij} \leq C_i \quad \forall i \in I \notag \\
	 &x_{ij} \geq 0 \quad \forall i, j  \notag \\
	 &x_{ij} = 0 \quad \forall (i,j) \notin E \notag.
  \end{align}

Our technique is as follows. 
We design a \textit{fractionally truthful} approximation algorithm which returns a feasible solution to \lpe. 
A fractionally truthful algorithm allocates fractional assignments to \bins, and no \bin\ can improve its fractional value by an untruthful report.
In particular, a fractionally truthful algorithm $\mathcal{A}^F$ takes $(T,E)$ and returns $x \in [0,1]^E$, a feasible solution to \lpe\ with the following property.
For each bin $i$, if the bin reports $E'_i \subset E_i$, we will have $\sum_{j:(i,j)\in E}v_{ij} x_{ij} \geq \sum_{j:(i,j)\in E}v_{ij}x'_{ij}$, where $x'= \mathcal{A}^F(T,E'_i\cup E_{-i})$ and $E_{-i}=E\setminus E_i$.
Next, we round the fractional solution using a special rounding technique which makes sure that each \bin\ obtains a fixed fraction of its fractional value in expectation. 
The \textit{randomized meta-rounding} is  capable of maintaining this fixed fraction. 

To use the randomized meta-rounding, we have to scale down the fractional solution by factor $2$, which is essentially the integrality gap of the \lpe\ \citep{shmoys1993approximation}.
Assuming $x^*=\mathcal{A}^F(T,E)$, the randomized meta-rounding represents $x^*/2$ as a convex combination of polynomially-many feasible integer solutions. 
Looking at the provided convex combination as a probability distribution over integer solutions, we sample a randomized solution $X$ which is always feasible, 
and its expected value is $1/2$ of the fractional value of $x^*$.
This is confirmed by Theorem \ref{thm-frac-int-sol} from the literature.

\begin{theorem} \citep{dughmi2010truthful} \label{thm-frac-int-sol}
If there exists a fractionally truthful $\alpha$-approximation algorithm for \gapms, then there exists a \tie\ $(2\alpha)$-approximation solution for \gapms.
\end{theorem}

\subsection{Multiple Knapsack Problem} \label{subsec:mkp}
We consider a variant of \gapms\ in which neither the size nor the value of each \job\ depends on the bins. Formally, for each \job\ $j$ we have $v_{ij}=v_j$ and $w_{ij}=w_j$ for all \bins\ $i$. 
First, we observe an algorithm that returns a (fractional) optimal solution to \lpe\ is not fractionally truthful. For more details, we provide Example \ref{mkp-example} relegated to the Appendix.

We propose Algorithm \ref{alg:mkp}.
We choose \bin\ $i$ in an arbitrary order and (fractionally) assign compatible \jobs\ to it according to the decreasing order of value densities of \jobs\ $v_{j}/w_{j}$ until the capacity of the \bin\ is exhausted or all compatible \jobs\ are exhausted. Then we proceed to the next \bin\ with remaining (fractional) \jobs. 

\RestyleAlgo{boxruled}
\begin{algorithm}
\SetAlgoNoLine


	1. Sort \jobs\ according to the decreasing order of value densities $v_{j}/w_{j}$, breaking ties arbitrarily.
	
2. \ForEach {\bin\ $i$ chosen in an arbitrary order} {
	for each unassigned (fractional) \job\ $j$ where $(i,j) \in E$ in the order defined above, fractionally assign as much of the \job\ to \bin\ $i$ until the \job\ is exhausted or the \bin\ is full.
}

\Return the resulting assignment $x$.
\caption{Multiple Knapsack Problem.}
\label{alg:mkp}
\end{algorithm}

Algorithm \ref{alg:mkp} is fractionally truthful. 
It is known that assigning \jobs\ according to decreasing order of value densities, when fractional assignments are allowed, produces the highest fractional value for the \bin. Since \bins\ wish to maximize their values, and the algorithm is aligned with this goal, the bins have no incentive to lie, and the algorithm is fractionally truthful.
For example, consider a bin with capacity $M \gg 1$, and two items $1$ and $2$ such that $v_1=1+\epsilon$, $w_1=1$, and $v_2=w_2=M$. Algorithm \ref{alg:mkp} assigns item $1$, and $\frac{M-1}{M}$-fraction of item $2$ to the bin, resulting in $1+\epsilon+\frac{M-1}{M}M=M+\epsilon$ value for the bin, the highest possible value. If fractional allocations were not allowed, only item $1$ could be assigned to the bin since the remaining capacity ($M-1$) is not sufficient to assign item $2$. Thus, in the absence of fractional allocations, the bin has incentives to hide its compatibility with item $1$ in order to obtain item $2$.

With regard to the total value, we show that Algorithm \ref{alg:mkp} returns a $2$-approximate fractional solution. 
We compare the outcome of the algorithm with the optimal solution to the LP formulation of the problem shown below.
  \begin{align}
    \text{Maximize}   & \textstyle \quad \sum_{i=1}^n \sum_{j=1}^m v_{j}x_{ij} \tag{MKP-LP[E]} \\
    \text{subject to} & \textstyle \quad \sum_{i=1}^n x_{ij} \leq 1, \ & \forall j \in J \notag \\
	 & \textstyle \quad  \sum_{j=1}^m w_{j}x_{ij} \leq C_i, \ & \forall i \in I \notag  \\
	 &\quad  x_{ij} \geq 0, \ & \forall i, j  \notag \\
	 &\quad  x_{ij} = 0, \ & \forall (i,j) \notin E.  \notag 
  \end{align}
Assignment $x$ computed by Algorithm \ref{alg:mkp} is a feasible assigment since each item is assigned only once, and the capacity of the bins are respected by the algorithm. Thus, $x$ belongs to the region of feasible solutions to MKP-LP[E].

\begin{lemma} \label{lem-mkp-ratio}
Algorithm \ref{alg:mkp} returns a $2$-approximation solution to MKP-LP[E].
\end{lemma}
\begin{proof}
We will construct a feasible dual solution with a value at most twice the value obtained by
the algorithm, then by calling the weak duality theorem, the claim will follow.
Assume $x$ is the outcome of Algorithm \ref{alg:mkp}. Using $x$ we can construct a feasible solution to the dual of MKP-LP[E] given below.
  \begin{align}
    \text{Minimize}   & \textstyle  \quad \sum_{j=1}^m p_j + \sum_{i=1}^n u_i C_i  \tag{MKP-LPD[E]}\\
     \text{subject to} & \textstyle \quad p_j+u_i w_{j} \geq v_{j}, \ & \forall (i,j) \in E \notag \\
	 &\quad  u_i \geq 0, \ & \forall i \notag \\
	 &\quad  p_j \geq 0, \ & \forall j. \notag
  \end{align}
Initially, let $p=\vec{0}$ and $u=\vec{0}$. If \job\ $j$ gets exhausted, set $p_j=v_{j}$. 
Furthermore, for all \textit{full} \bins\ $i$, set $u_i=v_j/w_j$, $j$ being the last \job\ (fractionally) assigned to $i$. 
We can observe that this satisfies the constraint corresponding to each edge $(i,j)$. 
In particular, if \bin\ $i$ is full, then for each $j$ incident on $i$, either $j$ gets exhausted with this assignment or does not. 
If $j$ is exhausted we have $p_j=v_j$ and therefore the constraint holds.
If $j$ is not exhausted, we have $v_j/w_j \leq u_i$ since \jobs\ are assigned in decreasing order of value density and thus the constraint holds. 
If \bin\ $i$ is not full, every \job\ $j$ which is assigned to it is exhausted by this assignment. 
That is we have $p_j=v_{j}$ and the constraint thus holds.
For every \job\ $j$ which is not assigned to the \bin\ but $(i,j) \in E$, we have $p_j = v_{j}$ since the \job\ is exhausted due to another assignment.
In sum, we have constructed a feasible dual solution using $x$.

Now, we bound the value of the dual solution with respect to the primal solution. First, we observe that $\sum_{i,j} v_{j}x_{ij} \geq \sum_j p_j \sum_i x_{ij}$, since $p_j=v_j$ if $j$ is fully exhausted and $p_j=0$, otherwise. Second, $\sum_{i,j} v_{j}x_{ij}=\sum_i \sum_j \frac{v_{j}}{w_{j}} (w_{j}x_{ij}) \geq \sum_i u_i \sum_j (w_{ij}x_{ij})$, since if $x_{ij}>0$ then $v_j/w_j \geq u_i$. Therefore, we obtain

\begin{displaymath}
\begin{array}{ll}

2\sum_{i,j} v_{j}x_{ij} &\geq  \sum_j p_j \sum_i x_{ij} + \sum_i u_i \sum_j (w_{j}x_{ij}) \\
 & = \sum_j p_j + \sum_i u_i C_i

\end{array}
\end{displaymath}

Notice, only for \jobs\ $j$ which get exhausted ($\sum_i x_{ij}=1$) we have $p_j>0$ and only for full \bins\ ($\sum_j w_{j}x_{ij}=C_i$) we have $u_i>0$.
The final term is the value of the dual, the desired conclusion.
\end{proof}


Finally, we call Theorem \ref{thm-frac-int-sol} and obtain the following.

\begin{theorem}
There exists a \tie\ $4$-approximation mechanism for the multiple knapsack problem in our model.
\end{theorem}
 
\subsection{Truthful Mechanism for \gapms} \label{subsec:inv-val-dens}
Now, we attempt to design a truthful algorithm for \gapms, but first solve the problem with an additional assumption.
We assume that the \textit{value density} of each \job\ is the same over all \bins. 
More formally, there exists a value $d_j$ for each \job\ $j$ such that for all \bins\ $i$, we have $\frac{v_{ij}}{w_{ij}} = d_j$. 
This assumption will be relaxed in Subsection \ref{subsec-unequl-density}. 
We design a truthful $4$-approximation mechanism for \gapms\ under this extra assumption. 

The proposed algorithm can be viewed as a variant of the \emph{deferred acceptance algorithm} designed for matching marketplaces.
Each \job\ $j$ has a preference list $\mathcal{L}_j$ according to decreasing order of $v_{ij}$ where $(i,j) \in E$, breaking ties arbitrarily. 
The preference list of a \bin\ is defined according to the decreasing order of value densities.
Once a (fractional) \job\ and a \bin\ are matched, the assignment will never be broken.

\RestyleAlgo{boxruled}
\SetAlgoNoEnd
\begin{algorithm}
\SetAlgoNoLine

\KwData{Preference lists of the \jobs, $\{\mathcal{L}_j\}_{j}$.}
\KwResult{A feasible solution $x$ to \lpe.}

1. Sort \jobs\ according to their decreasing order of value densities $d_j$, breaking ties arbitrarily.

2. \ForEach {\job\ $j$ chosen according to the order above} {
Fractionally assign as much of the \job\ to the \bins\ chosen according to the order specified by $\mathcal{L}_j$, until the \job\ is exhausted or all the \bins\ in $\mathcal{L}_j$ are full.
}
\Return the resulting assignment $x$.

\caption{GAP with Equal Density}
\label{alg:equal-density}
\end{algorithm}
To show the approximation factor of the solution, we can construct a feasible dual solution whose value is at most twice the value obtained by Algorithm \ref{alg:equal-density}, then by calling the weak duality theorem, the following lemma holds.

\begin{lemma} \label{lem-equal-density-approx}
Algorithm \ref{alg:equal-density} returns a $2$-approximation solution to \lpe\ when each \job\ has the same value density over \bins.
\end{lemma}

The truthfulness proof of Algorithm \ref{alg:equal-density} proceeds as follows. 
We first show that in an instance with $2$ \jobs\ and $2$ \bins\ ($2 \times 2$), truthfulness holds. This instance contains the core of the truthfulness proof for the general case. Truthfulness for simpler cases is trivial. 
A straightforward generalization of the argument for $2\times 2$ shows truthfulness for settings with $m$ \jobs\ and $2$ \bins\ ($2\times m$) for any $m>2$. For the general case of ($n \times m$) we provide an inductive argument.

In order to show that Algorithm \ref{alg:equal-density} is \textit{fractionally truthful}, we look at Algorithm \ref{alg:equal-density} as a variant of the deferred acceptance algorithm where \jobs\ propose capacities to \bins. 
In Step $2$ of the algorithm, we process items one by one.
For each item, we try to assign the item or part of the item to the bins according to the decreasing order of the item values for the bins.
To simplify the exposition of the proof, we view this process as items proposing to the bins. 
When \bin\ $i$ reveals its compatibility with item $j$, we view it as \bin\ $i$ accepting (possible) proposal by item $j$ as far as the capacity of the bin permits.
Similarly, \bin\ $i$ hiding its compatibility with item $j$ can be viewed as \bin\ $i$ rejecting (possible) proposals by item $j$, or equivalently not allowing item $j$ to propose to bin $i$.

\begin{lemma} \label{lem-2-2-truthfulness}
Algorithm \ref{alg:equal-density} is fractionally truthful for $2 \times 2$ settings.
\end{lemma}

\begin{proof}
Let \b{1}, and \b{2} denote the \bins, and \b{p} and \b{q} denote the \jobs. 
Let us assume \textbf{p} precedes \textbf{q} in proposing to the \bins, i.e. $d_p \geq d_q$. 
Fix this order of proposing \jobs\ as well as the reports by \bin\ \b{2}.
We argue that bin \b{1} is never better off by hiding some of its edges $E_1$.
Showing the truthfulness for bin \b{2} is analogue.

Assume $(\b{1},\b{q}) \in E_1$. Then \bin\ \b{1} may receive a proposal from \b{q}, but obviously the \bin\ receives no proposal from \b{q} if the \bin\ reports $(\b{1},\b{q}) \notin E_1$. Thus, hiding compatibility with \b{q} might only make a loss for the \bin.

Now, we analyze the behavior of the algorithm for a similar change in report for \job\ \textbf{p}.
We need to show that when $(\b{1},\b{p}) \in E_1$ (case \rom{1}) the obtained value by the bin is at least as good as when $(\b{1},\b{p}) \notin E_1$ (case \rom{2}).
Then we conclude that when truely $(\b{1},\b{p}) \in E_1$, the \bin\ has no incentive to report $(\b{1},\b{p}) \notin E_1$.

In case \rom{1}, if only a fraction of the proposal by \b{p} is accepted by the \bin, then the \bin\ has become full by accepting a fraction of \b{p} (recall \textbf{p} precedes \textbf{q} in proposing to the \bins). 
Thus the obtained value by the \bin\ is maximum and can't be better off in case \rom{2}.
If in case \rom{1} no fraction of the proposal by \b{p} is accepted by \bin\ \b{1}, or if there is no proposal by \b{p} then there will be no improvement in the value of the bin in case \rom{2}, as well.
What remains is to show that the \bin\ cannot be better off in case \rom{2} when it accepts the proposal by \b{p} fully in case \rom{1}.

\begin{figure}[!htb]
    \centering

\begin{tikzpicture}[node distance = 15mm]
 \node [circle, draw, minimum size=0.8cm] (p) {\b{p}};
 \node [circle, draw, minimum size=0.8cm] (q) [below= of p] {\b{q}};

 \node [rectangle, draw, minimum size=1cm] (1) [right=of p] {\b{1}};
 \node [rectangle, draw, minimum size=1cm] (2) [right= of q] {\b{2}};

 \draw [<-] (1) to  
  node [xshift=7mm, yshift=2mm, very near end] {$C_{1p}$} (p);
 \draw [-] (2) 
 to node [xshift=3mm, yshift=0mm, very near end] {$0$} (p);
 \draw [<-] (1) to 
  node [xshift=0mm, yshift=4mm, very near end] {$C_{1q}$} (q);
 \draw [<-] (2) to 
  node [xshift=3mm, yshift=2mm, very near end] {$C_{2q}$} (q); 
\end{tikzpicture}

        (a) Case \rom{1}. \b{p} is exhausted when it is assigned to \b{1}. 

\begin{tikzpicture} [node distance = 15mm]
 \node [circle, draw, minimum size=0.8cm] (p) {\b{p}};
 \node [circle, draw, minimum size=0.8cm] (q) [below= of p] {\b{q}};

 \node [rectangle, draw, minimum size=1cm] (1) [right=of p] {\b{1}};
 \node [rectangle, draw, minimum size=1cm] (2) [right= of q] {\b{2}}; 
 
 \draw [<-] (2) to 
 node [xshift=3mm, yshift=1mm, very near end] {$C_{2p}$} (p);
 \draw [<-] (1) to   
 node [xshift=0mm, yshift=4mm, very near end] {$C'_{1q}$} (q);
 \draw [<-] (2) to
  node [xshift=3mm, yshift=2mm, very near end] {$C'_{2q}$} (q); 
\end{tikzpicture}

        (b) Case \rom{2}. \b{1} hides its compatibility with \b{p}. At least a fraction of \b{q} is assigned to \b{1} ($C'_{1q}>0$). 

        \caption{Two cases where the \bin\ is and is not on the preference list of the \job. The amount of proposed and accepted capacities are shown on the edges.}
            \label{fig-gap-base-example}
\end{figure}

This situation is depicted in Figure \ref{fig-gap-base-example}. 
In the figure, (a) and (b) correspond to case \rom{1} and case \rom{2}, respectively.
In the figure, $C_{ij}$ denotes the capacity proposed by \job\ $j$ to \bin\ $i$, which is accepted by the bin. 
Considering the information provided in Fig. \ref{fig-gap-base-example}, we need to show that $C'_{1q} \leq C_{1q}+C_{1p}$. This will mean, in case \rom{2}, the \bin\ actually receives less capacity from \jobs\ with less (or equal) value densities than in case \rom{1}, which in turn means a lower value for \bin\ \b{1}. 
Notice, to arrive at this inequality we used the assumption that the order of proposing \jobs\ is fixed in the two setups.
To show the inequality, we first observe two facts about Algorithm \ref{alg:equal-density}.

\begin{observation} \label{obs:inc-capacity}
If a set of \jobs\ together propose a capacity of $C^0\leq C$ to a bin with capacity $C$, the bin will accept the whole proposed capacity.
If we first let a capacity $C^1$ propose to the \bin\, and afterwards let the foregoing \jobs\ propose the capacity $C^0$, the \bin\ will reject a capacity of at most $C^1$ from the \jobs\ that propose after the first capacity.
\end{observation}
\begin{proof}
Assume $C^1$ and $C^0$ in order propose to the bin.
If the \bin\ gets full by accepting $C^1$ we must have $C^1\geq C$, then the \bin\ will reject exactly a capacity of $C^0$ of the next \jobs. 
Now because $C^0 \leq C \leq C^1$, the claim holds.
If not (the \bin\ still has an empty capacity of $C^E$ after accepting $C^1$), the \bin\ accepts $C^1$ fully and rejects an amount equal to $\max \{0,C^0-C^E\}$ from the next proposing capacities. 
We have $C^1+C^E=C \geq C^0$, therefore $C^0-C^E \leq C^1$.
Thus, in this case the rejected capacity will be upper bounded by $C^1$.
This completes the proof.
\end{proof}

\begin{observation} \label{obs:up-capacity}
Let \b{1} and \b{2} be two subsequent \bins\ in $\mathcal{L}_j$. If \bin\ \b{1} rejects the proposed capacity $C_{1j}$ by \job\ $j$ then, this is an upper bound to $C_{2j}$, the capacity that will be proposed by \job\ $j$ to \b{2}, i.e. $C_{1j} \geq C_{2j}$.
\end{observation}
\begin{proof}
First, we must have $w_{1j} \geq w_{2j}$ since $\frac{v_{1j}}{w_{1j}}=\frac{v_{2j}}{w_{2j}}$ by the assumption of equal density over \bins\ and $v_{1j}\geq v_{2j}$ as \b{1} precedes \b{2} in $\mathcal{L}_j$.
Rejecting $C_{1j}$ means that this fraction of the \job\ remains: $C_{1j}/w_{1j}$. Then what will be proposed to \b{2} is $C_{2j}=w_{2j} \cdot (C_{1j}/w_{1j}) \leq C_{1j}$.
This completes the proof.
\end{proof}

Back to the argument about cases \rom{1} and \rom{2}, we notice that in case \rom{2} there is an increase of amount $C_{2p}$ in the proposed capacity to \b{2} compared to case \rom{1}. 
The capacity rejected by bin \b{2} is thus upper bounded by $C_{2p}$ according to Observation \ref{obs:inc-capacity}. 
That means, $C_{2q}-C'_{2q} \leq C_{2p}$. 
Moreover, according to Observation \ref{obs:up-capacity}, the rejected capacity upper bounds the proposed capacity to the next \bin. 
Hence, we have $C'_{1q}-C_{1q} \leq C_{2q}-C'_{2q}$. 
Therefore, we obtain $C'_{1q} \leq C_{1q}+C_{2p} \leq C_{1q}+C_{1p}$. 
The last inequality holds again because of Observation \ref{obs:up-capacity} 
(see it as bin \b{1} rejecting $C_{1p}$, an upper bound to $C_{2p}$).
This completes the proof of Lemma \ref{lem-2-2-truthfulness}.
\end{proof}


A simple generalization of the argument for $2\times 2$ markets shows truthfulness for the $2 \times m$ markets with $m>2$. 
A useful observation here is that we only need to show that \bin\ \b{1} will always report $E_1$ rather than $E_1\setminus \{e_j\}$ for every $e_j \in E_1$ . 
If we show this, we have in fact shown that reporting $E_1$ is better than reporting $E_1\setminus \{e_j\}$. 
This also shows that reporting $E_1\setminus \{e_j\}$ is better than hiding one edge from $E_1\setminus \{e_j\}$, i.e. reporting $E_1\setminus \{e_j, e_{j'}\}$ and so on.
For the general case we provide an inductive argument.
We assume that in a $(n-1) \times m$ setting \bins\ are truthful and prove that in a $n \times m$ setting truthfulness holds as well.

\begin{lemma} \label{lem-ind-step-truthfulness}
If Algorithm \ref{alg:equal-density} is truthful for markets with $m$ \jobs\ and $n-1$ \bins, it will be truthful for $n \times m$ markets for $n\geq 3$, and $m\geq 2$.
\end{lemma}
\begin{proof}
Consider \bin\ \b{i} and fix the reports of other \bins\ denoted by \b{-i}. 
We assume $($\b{i},\b{p}$)$ $\in E_i$ (case \rom{1}) and show that the \bin\ will never be better off by reporting $($\b{i},\b{p}$)$ $\notin E_i$ (case \rom{2}).
We compare the utility of the \bin\ in the two cases under a fixed order of proposing \jobs. 
The two cases are depicted in Figure \ref{fig:gap-induction}. 
Since the \jobs\ before \b{p} are assigned similarly in both cases, we only consider the \jobs\ which are processed after \b{p} denoted by \b{-p}.
\begin{figure}[!htb]
    \centering

\begin{tikzpicture} [node distance = 15mm]
 \node [circle, draw, minimum size=0.8cm] (p) {\b{p}};
 \node [circle, draw, minimum size=0.8cm] (m) [below= of p] {\b{-p}};

 \node [rectangle, draw, minimum size=1cm] (1) [right=of p] {\b{i}};
 \node [rectangle, draw, minimum size=1cm] (2) [right= of m] {\b{-i}}; 
 
 \draw [<-] (1) to  
  node [xshift=7mm, yshift=2mm, very near end] {$C_{ip}$} (p);
 \draw [-] (2) to node [auto] {} (p);
 \draw [<-] (1) to   
 node [xshift=0mm, yshift=5mm, very near end] {$C_{i,-p}$} (m);
 \draw [<-] (2) to
  node [xshift=3mm, yshift=2mm, very near end] {$C_{-i,-p}$} (m); 
\end{tikzpicture}

        (a) Case \rom{1}. \b{p} is exhausted when it is assigned to \b{i}. \b{i} may get a fraction or nothing from other items \b{-p}.


\begin{tikzpicture} [node distance = 15mm]
 \node [circle, draw, minimum size=0.8cm] (p) {\b{p}};
 \node [circle, draw, minimum size=0.8cm] (m) [below= of p] {\b{-p}};

 \node [rectangle, draw, minimum size=1cm] (1) [right=of p] {\b{i}};
 \node [rectangle, draw, minimum size=1cm] (2) [right= of m] {\b{-i}};

 \draw [<-] (2) to 
 node [xshift=4mm, yshift=1mm, very near end] {$C_{-i,p}$} (p);
 \draw [<-] (1) to 
  node [xshift=0mm, yshift=5mm, very near end] {$C'_{i,-p}$} (m);
 \draw [<-] (2) to 
  node [xshift=3mm, yshift=2mm, very near end] {$C'_{-i,-p}$} (m); 
\end{tikzpicture}

        (b) Case \rom{2}. \b{i} hides its compatibility with \b{p}. At least a fraction of \b{-p} is assigned to \b{i}. \b{p} is (fully) accepted by \b{-i}.

        \caption{Two cases where the \bin\ shows or hides its compatibility with an item.}
            \label{fig:gap-induction}
\end{figure}

We show that $C'_{i,-p} \leq C_{ip}+C_{i,-p}$, where $C_{i,-p}=\sum_{q \in -p}C_{i,q}$ and $C'_{i,-p}=\sum_{q \in -p}C'_{i,q}$. 
This means that \bin\ \b{i} in case \rom{2} actually receives less capacity from \jobs\ with less (or equal) value densities than case \rom{1}, which in turn implies lower value for the \bin.

Consider case \rom{2}. 
We look closer at the bin(s) to which \job\ \b{p} will be assigned. 
We assume \b{p} is (fractionally) assigned to at least one \bin\ otherwise we have $C_{i,-p}=C'_{i,-p}$ and thus the claim holds. Let \bin\ \b{1} be the first bin to which \b{p} will be assigned. 

We assume \bin\ \b{1} gets full at some point otherwise this \bin\ accepts the extra capacity ($C_{1p}$, the capacity proposed by item \b{p} to bin \b{1}) without rejecting any capacity and therefore we have $C_{i,-p}=C'_{i,-p}$ and thus the claim holds. 
When bin \b{1} gets full, some of the currently proposing \jobs\ to \bin\ \b{1} 
will stop proposing to it and go to the next \bin\ in their preference list. 
Let us call these capacities $C_1$. 
$C_1$ is upper bounded by $C_{1p}$ according to Observation \ref{obs:inc-capacity} which in turn is upper bounded by $C_{ip}$ based on Observation \ref{obs:up-capacity}: $C_1 \leq C_{1p} \leq C_{ip}$. 
If $C_1$ directly proposes to bin \b{i}, the bin won't be better off in case \rom{2} because $C_1 \leq C_{ip}$.
The situation is worse for bin \b{i}, if $C_1$ goes to the other bins.
One can view this situation as \bin\ \b{i} rejecting capacity $C_1$ in a $(n-1) \times m$ setting where \bin\ \b{1} (which is now full) and its absorbed capacities are eliminated. 
According to our induction assumption, this strategy will not make \bin\ \b{i} better off in a $(n-1) \times m$ setting. 
This completes the proof.
\end{proof}

Taking into account, Lemma \ref{lem-2-2-truthfulness} and Lemma \ref{lem-ind-step-truthfulness}, we obtain the following.
\begin{lemma} \label{lem-equal-density-ic}
Algorithm \ref{alg:equal-density} is fractionally truthful.
\end{lemma}

Finally, by calling Theorem \ref{thm-frac-int-sol}, we obtain the following.

\begin{theorem} \label{thm-equ-density}
There exists a \tie\ $4$-approximation mechanism for \gapms\ when each \job\ has the same value density over all \bins.
\end{theorem}

\subsection{Unequal Value Densities} \label{subsec-unequl-density}
We presented a \tie\ $4$-approximate mechanism for \gapms\ when each \job\ has a unique value density over all \bins. 
Now we explain how to relax this assumption at the expense of a logarithmic loss in the total value.
Consider those edges in $E$, $e=(k,l)$ and $e'=(k',l')$ whose value densities are respectively upper and lower bounds over all value densities:
\begin{displaymath}
L=\frac{v_{k'l'}}{w_{k'l'}} \leq \frac{v_{ij}}{w_{ij}} \leq \frac{v_{kl}}{w_{kl}}=U, \hspace{5pt} \forall (i,j) \in E.
\end{displaymath}
Let us assume $U$ and $L$ are publicly known.
This assumption will be removed later. 
Knowing this information we choose a density value $d$ uniformly at random from the set $D=\{U, \frac{U}{2}, \frac{U}{4}, \ldots, \frac{U}{2^{O(\ln{(U/L)})}}\}$.
Then we define a new valuation $\hat{v}$ as follows.
For every edge $(i,j)$ in $E$ with $\frac{v_{ij}}{w_{ij}} < d$ we set $\hat{v}_{ij}=0$, or equivalently the edge is discarded from the graph. 
For every $\frac{v_{ij}}{w_{ij}} \geq d$, define $\hat{v}_{ij}$ such that $\frac{\hat{v}_{ij}}{w_{ij}}=d$. 
Notice that always $\hat{v}_{ij} \leq v_{ij}$.
Now we have an instance of \gapms\ with equal densities for which there exists a \tie\ 4-approximate mechanism according to Theorem \ref{thm-equ-density}. 
To ensure truthfulness, in the end, if \job\ $j$ is assigned to \bin\ $i$ by the subroutine for equal value densities, we withdraw the \job\ with probability $1-\frac{\hat{v}_{ij}}{v_{ij}}$. 
In other words, we let the \bin\ hold the \job\ with probability $\frac{\hat{v}_{ij}}{v_{ij}}$. 
If item $j$ is assigned to bin $i$, the generated value for the bin will be $v_{ij}$, but if we let the bin hold the item with probability $\frac{\hat{v}_{ij}}{v_{ij}}$, then the expected value will be $\hat{v}_{ij}$.
This way, we make sure that each \job\ has the same value density over all \bins\ as it is required by the subroutine to guarantee truthfulness.

Set $D$ contains $O(\ln{(U/L)})$ densities, and each density has the probability of $p=\frac{1}{O(\ln{(U/L)})}$ to be chosen. 
At least half of each valuation $v_{ij}$ with probability $p$ is counted in the expected total value; therefore, we obtain an $O(\ln{(U/L)})$ approximation factor. 

To remove the assumption that $U$ and $L$ are public, we certify that the bins $k$ and $k'$ wouldn't hide the corresponding edges.
To this end, we run one of the following three algorithms with probability $1/3$.
$i)$ Let bin $k$ (the owner of edge $e$) choose all its desired items and assign nothing to the other bins.
$ii)$ Let bin $k'$ (the owner of $e'$) choose all its desired items and assign nothing to the other bins.
$iii)$ Exclude bin $k$ and $k'$ and run the algorithm above for all other bins using $U$ and $L$ obtained from the two excluded bins.
One can observe that the two bins $k$ and $k'$ cannot do any better by hiding their edges.
Also, it is easy to observe that the approximation factor is still $O(\ln{(U/L)})$. Thus, we obtain the following.

\begin{theorem}
There exists a \tie\ $O(\ln{(U/L)})$ approximate mechanism for \gapms.
\end{theorem}

We leave open the question of whether there exists a truthful mechanism with a constant factor of approximation for \gapms.

\vspace{3pt}
\hspace{-0.7cm}{\bf Acknowledgement}

The authors gratefully acknowledge support of the Deutsche Forschungsgemeinschaft (DFG) (BI 1057/7-1).
The authors also would like to acknowledge the
anonymous referees for the their useful comments and suggestions.

  \bibliographystyle{elsarticle-num} 
  \bibliography{literature}

\appendix

\section{}

Proof of Theorem \ref{thm-low-bounds} 
\emph{No truthful deterministic algorithm with an approximation ratio better than $2$ exists for \gapms.
Moreover, no truthful-in-expectation randomized algorithm with an approximation ratio better than $1.09$ exists for \gapms.}

\begin{proof}
Consider a small market with two bins and two items shown in Figure \ref{fig-gap-wo-money-low-bound} (a).
In this market, bins have capacity $1$, and are both compatible with the two items.
Item $B$ is more valuable to both bins ($x>1$), but each item has size $1$.
This market can be viewed as an instance for both the multiple-knapsack problem (Subsection \ref{subsec:mkp}), and the density-invariant GAP (Subsection \ref{subsec:inv-val-dens}).

Regarding deterministic mechanisms, an arbitrary truthful mechanism has to assign item $B$ to one bin.
Without loss of generality, assume $B$ is assigned to bin \b{2}, i.e., the tie is broken deterministically (alphabetically) in favor of bin \b{2}.
Now, consider reports in (b) of Figure \ref{fig-gap-wo-money-low-bound}.
The mechanism, in case (b), cannot assign $B$ to bin \b{1} as it violates truthfulness.
Thus, the mechanism assigns $B$ to \b{2}, and this results in an approximation ratio of $\frac{x+1}{x}$ which tends to $2$ when $x$ gets very close to $1$.

An arbitrary truthful-in-expectation mechanism, in case (a), assigns $B$ to one bin with a probability less than or equal $1/2$.
Without loss of generality, let bin \b{1} be that bin.
The utility of bin \b{1}, in this case, will be at most $\frac{x+1}{2}$ for $x>1$.
Assume, in case (b), item $B$ is assigned to bin \b{1} with probability $q$, resulting a $q\cdot x$ expected value for bin \b{1}.
Truthfulness stipulates no increase in the utility of bin \b{1} in case (b), i.e., $\frac{x+1}{2} \geq q\cdot x$, thus $q \leq \frac{x+1}{2x}$.
In case (b), the total expected value will be $q(1+x)-(1-q)x = x+q$, thus the approximation ratio will be $\frac{x+1}{x+q}$.
In order to obtain a smaller approximation ratio, we plug in $q=\frac{x+1}{2x}$.
The ratio gets a value of $1+\frac{1}{4\sqrt{2}+5} \approx 1.094$ for $x=1+\sqrt{2}$, the desired conclusion.

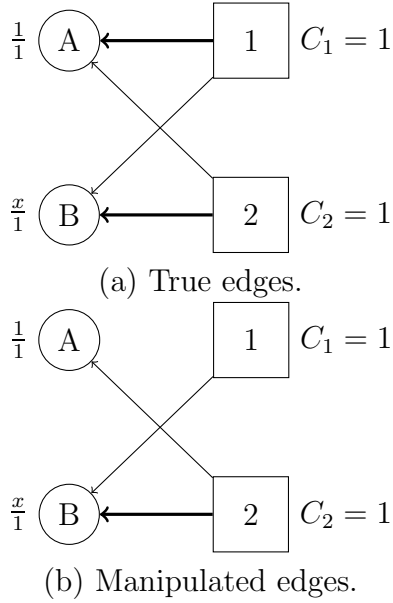
\begin{figure}[!htb]
    \centering

\begin{tikzpicture} [node distance = 15mm]
 \node [circle, draw, minimum size=0.8cm] (A) {A};
 \node [left=0mm of A]{$\frac{1}{1}$};
 \node [circle, draw, minimum size=0.8cm] (B) [below= of A] {B};
 \node [left=0mm of B]{$\frac{x}{1}$};

 \node [rectangle, draw, minimum size=1cm] (1) [right=of A] {1};
 \node [rectangle, draw, minimum size=1cm] (2) [right= of B] {2}; 
 \node [right=0mm of 1]{$C_1=1$}; 
 \node [right=0mm of 2]{$C_2=1$}; 
 
 \draw [->, very thick] (1) to node [auto] {} (A);
 \draw [->] (1) to (B);
 \draw [->, very thick] (2) to (B); 
 \draw [->] (2) to (A);
\end{tikzpicture}

        (a) True edges.

\begin{tikzpicture} [node distance = 15mm]
 \node [circle, draw, minimum size=0.8cm] (A) {A};
 \node [left=0mm of A]{$\frac{1}{1}$};
 \node [circle, draw, minimum size=0.8cm] (B) [below= of A] {B};
 \node [left=0mm of B]{$\frac{x}{1}$};

 \node [rectangle, draw, minimum size=1cm] (1) [right=of A] {1};
 \node [rectangle, draw, minimum size=1cm] (2) [right= of B] {2}; 
 \node [right=0mm of 1]{$C_1=1$}; 
 \node [right=0mm of 2]{$C_2=1$}; 
 
 \draw [->] (1) to (B);
 \draw [->, very thick] (2) to (B); 
 \draw [->] (2) to (A);
\end{tikzpicture}

        (b) Manipulated edges.

        \caption{Circles represent \jobs\ and squares represent \bins. The $\frac{\text{value}}{\text{size}}$ of each \job\ is on its left. 
        Each bin has a capacity of $1$. Selected assignments are in bold.}
            \label{fig-gap-wo-money-low-bound}
\end{figure}

\end{proof}

Proof of Theorem \ref{thm-frac-int-sol}.
\emph{If there exists a fractionally truthful $\alpha$-approximation algorithm for \gapms, then there exists a \tie\ $(2\alpha)$-approximation solution for \gapms.}
\begin{proof}
Let $\mathcal{A}^F$ denote a fractionally truthful algorithm for \gapms\ that takes an instance $(T,E)$ and returns a feasible solution to \lpe.
Let $x^*$ be the outcome of $\mathcal{A}^F$ on instance $(T,E)$.
Let $\{X^l\}_{l\in L}$ denote the set of feasible integer solutions to \lpe, where $L$ indexes all feasible integer solutions.
The  integrality gap of \lpe\ equals $2$ \citep{shmoys1993approximation}, thus we scale down the fractional solution by factor $2$.
The meta-randomized rounding applied to $x^*/2$ computes a probability distribution over feasible integer solutions whose support is polynomial \citep{carr2000randomized, lavi2011truthful}:
\begin{displaymath}
\begin{array}{lll}
\frac{x^*}{2}=\sum_{l\in L} \lambda_l X^l, &\ \sum_{l\in L} \lambda_l=1,\ \text{and} &\ \forall l\in L, \lambda_l\geq 0.
\end{array}
\end{displaymath}

We treat the convex decomposition above as a probability distribution according to which solution $X^l$ has probability $\lambda_l$ of being selected.
Let $X$ be a solution sampled from the above distribution.
Obviously $X$ is feasible by the construction of the distribution.
We also have $\E[X_{ij}]=\frac{1}{2}x^*_{ij}$ for all $i$ and $j$ from the construction of the distribution.
By the linearity of expectation, the expected value of a bin is $\E[\sum_{j:(i,j)\in E}v_{ij} X_{ij}] = \frac{1}{2}\sum_{j:(i,j)\in E}v_{ij} x^*_{ij}$.
Therefore, the expected value of the solution is exactly half of the value of the fractional solution.

For truthfulness, fix bin $i$ and $E_{-i}=E\setminus E_i$.
Suppose the bin reports $E'_i \subset E_i$ rather than $E_i$.
Let $x'= \mathcal{A}^F(T, E'_i \cup E_{-i})$, and $X'$ be the solution returned by the meta-randomized rounding from $x'/2$.
We have
\begin{displaymath}
\begin{array}{ll}
\E[\sum_{j:(i,j)\in E}v_{ij} X_{ij}] &= \frac{1}{2}\sum_{j:(i,j)\in E}v_{ij} x^*_{ij} \\
&\geq \frac{1}{2}\sum_{j:(i,j)\in E}v_{ij} x'_{ij} \\
&=\E[\sum_{j:(i,j)\in E}v_{ij} X'_{ij}]
\end{array}
\end{displaymath}
The inequality is because $\mathcal{A}^F$ is fractionally truthful.
Therefore, the bin cannot improve its expected value by hiding some of its edges.
This completes the proof.
\end{proof}


Proof of Lemma \ref{lem-equal-density-approx}.
\emph{Algorithm \ref{alg:equal-density} returns a $2$-approximation solution to \lpe\ when each \job\ has the same value density over \bins.}
\begin{proof}
An argument similar to that of Lemma \ref{lem-mkp-ratio} in addition to some required modifications will show the claim.
Assume $x$ is the outcome of Algorithm \ref{alg:equal-density}. Using $x$ we can construct a feasible solution to the dual of LP[E] (LPD[E] given below) which is not greater than twice the value of $x$. Then we call the weak LP-duality theorem and conclude that $x$ is a $2$ approximate solution to LP[E].
  \begin{alignat*}{2}
    \text{LPD[E]:} & \\
    \text{Minimize}   & \textstyle \quad \sum_{j=1}^m p_j + \sum_{i=1}^n u_i C_i  \\
     \text{subject to} & \textstyle \quad p_j+u_i w_{ij} \geq v_{ij}, \ & \forall (i,j) \in E \\
	 & u_i \geq 0, \ & \forall i \\
	 & p_j \geq 0, \ & \forall j.\\
  \end{alignat*}
 
Initially, let $p=\vec{0}$ and $u=\vec{0}$. 
If \job\ $j$ gets exhausted when assigned to \bin\ $i$, set $p_j=v_{ij}$. 
Furthermore, for all \textit{full} \bins\ $i$, set $u_i=d_j$, $j$ being the last \job\ (fractionally) assigned to $i$. 
We can observe that this satisfies the constraint corresponding to each edge $(i,j)$. 
In particular, if \bin\ $i$ is full, then for each $j$ incident on $i$, $j$ either gets exhausted with this assignment or does not. 
If $j$ is exhausted, we have $p_j=v_{ij}$ and therefore the constraint holds.
If $j$ is not exhausted, we have $v_{ij}/w_{ij}=d_j \leq u_i$ since \jobs\ are assigned in decreasing order of value density and thus the constraint holds. 
If \bin\ $i$ is not full, every \job\ $j$ which is assigned to it is exhausted by this assignment. 
That is we have $p_j=v_{ij}$ and the constraint thus holds. 
For every \job\ $j$ which is not assigned to the \bin\ but $(i,j) \in E$, we have $p_j \geq v_{ij}$ since the \job\ is exhausted due to an assignment $(i',j) \in E$ with $v_{i'j}\geq v_{ij}$.
Therefore, we have constructed a feasible dual solution using $x$.

Now, we bound the value of the dual solution with respect to the primal solution. 
First, we observe that $\sum_{i,j} v_{ij}x_{ij} \geq \sum_j p_j \sum_i x_{ij}$, since $p_j$ lower bounds the value of any edge on which any part of \job\ $j$ is assigned ($x_{ij}>0$) because the \job\ goes to \bins\ according to the order specified by $\mathcal{L}_j$. 
Second, $\sum_{i,j} v_{ij}x_{ij}=\sum_i \sum_j \frac{v_{ij}}{w_{ij}} (w_{ij}x_{ij})\geq \sum_i u_i \sum_j (w_{ij}x_{ij})$, since if $x_{ij}>0$ then $\frac{v_{ij}}{w_{ij}}= d_j \geq u_i$. 
Therefore, we obtain

\begin{displaymath}
\begin{array}{ll}

2\sum_{i,j} v_{ij}x_{ij} &\geq  \sum_j p_j \sum_i x_{ij} + \sum_i u_i \sum_j (w_{ij}x_{ij}) \\
 & = \sum_j p_j + \sum_i u_i C_i

\end{array}
\end{displaymath}

Notice, only for \job\ $j$ which gets exhausted ($\sum_i x_{ij}=1$), we have $p_j>0$ and only for full \bins\ ($\sum_j w_{j}x_{ij}=C_i$) we have $u_i>0$.
The final term is the value of the dual, the desired conclusion.
\end{proof}

\begin{example}[Multiple Knapsack Example]\label{mkp-example}
we observe an algorithm that returns a (fractional) optimal solution to \lpe\ is not fractionally truthful for the multiple knapsack problem.  
This can be seen in the example shown in Figure \ref{fig-mkp-counter-example}.
In (a) of this figure, the edges are reported truthfully, and the value-maximizing allocation, assigns $A$ to \bin\ \b{1} and the other item to the other bin.
In (b) of this figure, \bin\ \b{1} hides its compatibility with item A and as a consequence it is better off (in expectation) when the mechanism maximizes the total value. 
In (b) of this figure, the tie can be broken randomly or deterministically (alphabetically) in favor of \bin\ \b{1}. In any case, \bin\ \b{1} is better off by manipulation.
  
\begin{figure}[!htb]
    \centering

\begin{tikzpicture} [node distance = 15mm]
 \node [circle, draw, minimum size=0.8cm] (A) {A};
 \node [left=0mm of A]{$\frac{1+\varepsilon}{1}$};
 \node [circle, draw, minimum size=0.8cm] (B) [below= of A] {};
 \node [left=0mm of B]{$\frac{x}{1}$};

 \node [rectangle, draw, minimum size=1cm] (1) [right=of A] {1};
 \node [rectangle, draw, minimum size=1cm] (2) [right= of B] {}; 
 \node [right=0mm of 1]{$C_1=1$}; 
 \node [right=0mm of 2]{$C_2=1$}; 
 
 \draw [->, very thick] (1) to node [auto] {} (A);
 \draw [->] (1) to (B);
 \draw [->, very thick] (2) to (B); 
\end{tikzpicture}

        (a) True edges.

\begin{tikzpicture} [node distance = 15mm]
 \node [circle, draw, minimum size=0.8cm] (A) {A};
 \node [left=0mm of A]{$\frac{1+\varepsilon}{1}$};
 \node [circle, draw, minimum size=0.8cm] (B) [below= of A] {};
 \node [left=0mm of B]{$\frac{x}{1}$};

 \node [rectangle, draw, minimum size=1cm] (1) [right=of A] {1};
 \node [rectangle, draw, minimum size=1cm] (2) [right= of B] {}; 
 \node [right=0mm of 1]{$C_1=1$}; 
 \node [right=0mm of 2]{$C_2=1$}; 
 
 \draw [->, very thick] (1) to (B);
 \draw [->, very thick] (2) to (B); 
\end{tikzpicture}

        (b) Manipulated edges.

        \caption{Circles represent \jobs\ and squares represent \bins. The value/size of each \job\ is on its left. Value maximizing assignments are in bold. $x \gg 1$.}
            \label{fig-mkp-counter-example}
\end{figure}

\end{example}

\end{document}